\documentclass[twocolumn]{quantumarticle}

\usepackage{amsmath}
\usepackage{amsfonts}
\usepackage{amsthm}
\usepackage{amssymb}
\usepackage{url}
\usepackage{hyperref}
\usepackage{color}

\newcommand{\ket}[1]{\left | \, #1 \right \rangle}
\newcommand{\bra}[1]{\left \langle{} #1 \, \right |}

\newcommand{\norm}[1]{\left|\left|#1\right|\right|}

\newtheorem{theorem}{Theorem}

\newtheorem{lemma}{Lemma}
\newtheorem{corollary}{Corollary}

\begin{document}
%%%%%%%%%%%%%%%%%%%%%%%%%%%%%%%%%%%%%%%%%%%%%%%%%%%%%%%%%%%%%%%%%%%%%%%%%%%%%%%%%%%%%%%%%%%%%%%%%
\title{Self-testing in parallel with CHSH}
% %%%%%%%%%%%%%%%%%%%%%%%%%%%%%%%%%%%%%%%%%%%%%%%%%%%%%%%%%%%%%%%%%%%%%%%%%%%%%%%%%%%%%%%%%%%%%%%%%%
\author{Matthew McKague}
\affiliation{Department of Electrical Engineering and Computer Science, Queensland University of Technology}
\email{m.mckague@qut.edu.au}

\maketitle
\begin{abstract}
Self-testing allows classical referees to verify the quantum behaviour of some untrusted devices. Recently we developed a framework for building large self-tests by repeating a smaller self-test many times in parallel.  However, the framework did not apply to the CHSH test, which tests a maximally entangled pair of qubits. CHSH is the most well known and widely used test of this type.  Here we extend the parallel self-testing framework to build parallel CHSH self-tests for any number of pairs of maximally entangled qubits.  Our construction achieves an error bound which is polynomial in the number of tested qubit pairs.
\end{abstract}

%%%%%%%%%%%%%%%%%%%%%%%%%%%%%%%%%%%%%%%%%%%%%%%%%%%%%%%%%%%%%%%%%%%%%%%%%%%%%%%%%%%%%%%%%%%%%%%%%
\section{Introduction}
The most basic self-testing problem is this: given two non-communicating quantum devices, verify through classical interaction alone that they share a maximally entangled pair of qubits.  This problem can also be phrased in the language of non-local games.  In a non-local game two non-communicating players\footnote{Non-local games for more than two players are also possible.}, Alice and Bob, receive questions from a referee and attempt to win the game by providing answers which, together with the questions, satisfy some predicate.  For non-local games, self-testing says that for Alice and Bob to approach the optimal quantum probability of winning, their strategy must approach some particular ideal quantum strategy.

While it is interesting that it is possible to self-test a single pair of maximally entangled qubits, many more applications are possible if many pairs can be tested.  This motivates the idea of repeated self-tests, whether in series (as in~\cite{Reichardt2013:Classicalcommandofquantumsystems}) or parallel (as in~\cite{McKague:2015:Selftestingin}), to test many pairs of maximally entangled qubits.  From the players' perspective, strategies for serial repetition can be converted to parallel repetition but not the other way around.  Hence parallel self-testing results are more general and their robustness bounds also apply to serial repetition.

The CHSH game~\cite{Clauser:1969:Proposed-Experi} is arguably the most well-known non-local game and has seen wide application.  It is distinguished by its simplicity, with only single-bit questions and answers for both Alice and Bob.  Because of this simplicity, the CHSH game is widely used, both in theory and experiment.  Hence it is a natural candidate for use with parallel repetition.

An important distinction to make is between self-tests where the number of players is kept at two (as in~\cite{Reichardt2013:Classicalcommandofquantumsystems}), with self-tests where two new players are added for each copy of the self-test, (as in~\cite{Magniez:2006:Self-testing-of}).  Every additional player introduces either additional assumptions, or restrictions on potential experiments.  Hence the two player scenario is more general and preferred.  We will be considering the two-player scenario here.

\subsection{Previous work}
The history of self-testing using CHSH goes back to Popescu and Rohrlich~\cite{Popescu:1992:Which-states-vi} who characterised the optimal quantum strategies for CHSH.\@  Although the terminology and applications of self-testing were later introduced by Mayers and Yao~\cite{Mayers:2004:Self-testing-qu}.  Robust proofs of self-testing with CHSH appeared with Bardyn et al.~\cite{Bardyn:2009:Device-independ} and McKague et al.~\cite{McKague:2012:Robust-self-tes}.

As a non-local game, parallel repetition of CHSH was considered by Cleve et al.~\cite{Cleve:2008:Perfect-Paralle}, although they did not consider self-testing.  Reichardt et al.~\cite{Reichardt2013:Classicalcommandofquantumsystems} looked at serial repetition of CHSH for self-testing.  The first parallel self-testing result using CHSH is due to Wu et al.~\cite{Wu2015:Deviceindependentself} who used semi-definite programming and numerical methods to give robust error bounds for self-testing with two simultaneous games of CHSH.\@  They also give an analytic proof of self-testing for the ideal case, again for two copies of CHSH.\@

Self-testing using parallel repetition with more than two repetitions first appeared in~\cite{McKague:2015:Selftestingin}, which used the Mayers-Yao test~\cite{Mayers:2004:Self-testing-qu} and another CHSH-like self-test.  However, the robustness bound of the latter test scaled exponentially with the amount of repetitions.  Self-testing using parallel repetition of CHSH was shown by Coladangelo\footnote{
We first became aware of Coladangelo's work after an early version of this manuscript was posted to the arXiv.
} in~\cite{Coladangelo:2016:Parallelselftesting}.

\subsection{Contributions}

Our main contribution here is to give a self-testing construction for parallel repetition of the CHSH game.  Informally, we prove that if Alice and Bob have a strategy for playing $\frac{n}{2}$ rounds\footnote{i.e.\ for self-testing a state on $n$ qubits} of CHSH\ in parallel that achieves an average CHSH\ value of at least $2 \sqrt{2} - \epsilon$ (i.e.\ they are $\epsilon$ close to the quantum bound) then their state is within $O(n^{\frac{9}{8}} \epsilon^{\frac{1}{8}})$ of $\frac{n}{2}$ pairs of maximally entangled qubits.

As a step in our construction we also show that it is possible to self-test many maximally entangled pairs of qubits using a modified parallel repetition of CHSH with only a logarithmic number of questions (measurement settings).  This is analogous to a similar result built on the Mayers-Yao test in~\cite{McKague:2015:Selftestingin}.

In order to achieve our construction we slightly generalise the parallel self-testing framework given in~\cite{McKague:2015:Selftestingin}.  Our construction is in some ways a generalisation of the analytic proof of self-testing for two copies of CHSH given in~\cite{Wu2015:Deviceindependentself}.

%%%%%%%%%%%%%%%%%%%%%%%%%%%%%%%%%%%%%%%%%%%%%%%%%%%%%%%%%%%%%%%%%%%%%%%%%%%%%%%%%%%%%%%%%%%%%%%%%
\section{A framework for self-testing qubits}
Before addressing our particular self-test, we will recall and slightly modify the self-testing framework introduced in~\cite{McKague:2015:Selftestingin}.

%%%%%%%%%%%%%%%%%%%%%%%%%%%%%%%%%%%%%%%%%%%%%%%%%%%%
\subsection{Technical preliminaries}

We define $1_k$ to be the $n$-bit string which is 1 in the $k$-th position and 0 everywhere else.  For $x$ an $n$-bit string, let $|x|$ be the number of 1's in $x$ (the Hamming weight).  Further, when $n$ is even, define $x_a$ to be the first $\frac{n}{2}$ bits of $x$. The string $x_b$ is the last $\frac{n}{2}$ bits of $x$.  Later, we will divide $x$ between Alice and Bob so $x_a$ represents $x$ on Alice's side, and $x_b$ is for Bob's side.  For strings $x$ and $y$ of the same length, $x \cdot y = \sum_{j} x_j y_j$ is the inner product of $x$ and $y$.  Finally, $xy$ is the concatenation of $x$ and $y$.

If we have operators $M_1 \dots M_n$ then for a bit string $t \in {\{0,1\}}^n$ we define
\begin{equation}
M^t := \prod_{k=1}^n M_k^{t_k}.
\end{equation}
The order of the product is important since we do not know whether the operators all commute.  Hence we will make the convention that the index increases from left to right.

%%%%%%%%%%%%%%%%%%%%%%%%%%%%%%%%%%%%%%%%%%%%%%%%%%%%
\subsection{Sufficient conditions for self-testing}

In~\cite{McKague:2015:Selftestingin} we developed a framework for self-testing states of many qubits, and in particular derived sufficient conditions for self-testing many pairs of maximally entangled qubits.  Unfortunately, the framework in~\cite{McKague:2015:Selftestingin} is not general enough for our purposes since the sufficient conditions require the existence of a large set of commuting operators, but we will be unable to provide them here.  In this section we will modify our previous work to provide a set of sufficient conditions which requires a smaller set of commuting operators.

To begin with, we recall Lemma~4 from~\cite{McKague:2015:Selftestingin}, which gives a very general set of sufficient conditions for self-testing a many qubit state.  Here we present a version which is restricted to many pairs of maximally entangled qubits, rather than the more general version in~\cite{McKague:2015:Selftestingin} which also applies to graph states.

\begin{lemma}[McKague~\cite{McKague:2015:Selftestingin}, Lemma 4]
\label{lemma:graphstateselftestconditions}
Let $\ket{\psi}$ be the $n$-qubit state
\begin{equation}
\ket{\psi} =
\frac{1}{\sqrt{2^n}}
\sum_{u_a, u_b \in {\{0,1\}}^{n/2} } {(-1)}^{u_a \cdot u_b} \ket{u}.
\end{equation}
Further suppose that $\ket{\psi^{\prime}}$ is a normalised state in a Hilbert space $\mathcal{H}_{\mathcal{A}}$, and ${\{X^{\prime}_{j}\}}_{j=1}^n$, ${\{Z^{\prime}_{j}\}}_{j=1}^n$ are unitary, Hermitian operators on $\mathcal{H}_{\mathcal{A}}$ such that for any $s,t \in {\{0,1\}}^{n}$
\begin{equation}
\label{eq:XZanticommute9}
	\norm{
		Z^{\prime t} X^{\prime s} \ket{\psi^{\prime}}
		-
		{(-1)}^{s \cdot t}
		X^{\prime s} Z^{\prime t} \ket{\psi^{\prime}}
	}
	\leq
    \epsilon
\end{equation}
and
\begin{equation}
\label{eq:XZswapgeneral9}
	\norm{
		Z^{\prime s_b s_a} \ket{\psi^{\prime}}
		-
		{(-1)}^{s_{a} \cdot s_b}
		X^{\prime s} \ket{\psi^{\prime}}
	}
	\leq
    \epsilon.
\end{equation}
Then there exists an isometry $\Phi$ and a state $\ket{junk}$ such that for any $p,q \in {(0,1)}^{n}$
\begin{equation}
    \norm{
        \Phi(
            X^{\prime q}
            Z^{\prime p}
            \ket{\psi^{\prime}}
        )
        -
        \ket{junk}
        X^q
        Z^p
        \ket{\psi}
    }
    \leq
    O(\sqrt{\epsilon}).
\end{equation}
\end{lemma}
The ideal state $\ket{\psi}$ is the graph state corresponding to $\frac{n}{2}$ isolated edges and is hence $\frac{n}{2}$ pairs of maximally entangled states.

The conditions for this lemma involve products of many operators.  In general we will not be able to directly deduce properties of such complex products.  Hence we will derive a less general set of conditions which are easier to meet from the limited information provided by measurement outcomes.  In~\cite{McKague:2015:Selftestingin} this step is provided by Lemma 5.  However, the conditions used are too restrictive to be applicable in our case.  Hence we will prove a less restrictive version of this lemma.

\begin{lemma}
\label{lemma:acandxz}
Suppose that $n$ is even and
\begin{enumerate}
\item $\ket{\psi^{\prime}} \in \mathcal{H}_A \otimes \mathcal{H}_B$ is a state
\item ${\{X^{\prime}_{k}\}}_{k=1}^{n/2}$  are pair-wise commuting operators on $\mathcal{H}_A$
\item ${\{Z^{\prime}_k\}}_{x=1}^{\frac{n}{2}}$  are pair-wise commuting operators on $\mathcal{H}_A$
\item ${\{X^{\prime}_k\}}_{k=\frac{n}{2} + 1}^{n}$  are operators on $\mathcal{H}_B$
\item ${\{Z^{\prime}_k\}}_{k=\frac{n}{2} + 1}^{n}$  are operators on $\mathcal{H}_B$
\end{enumerate}
such that for all $k \neq \ell$
\begin{eqnarray}
	\norm{X^{\prime}_k}_\infty , \norm{Z^{\prime}_k}_\infty & \leq & 1
\\
\label{eq:XZapproxcommute2old}
	\norm{
		X^{\prime}_k Z^{\prime}_\ell \ket{\psi^{\prime}}
		-
		Z^{\prime}_\ell X^{\prime}_k \ket{\psi^{\prime}}
	}
	& \leq & \epsilon_1
\\
\label{eq:XnearZ2}
	\norm{
		X^{\prime}_k \ket{\psi^{\prime}}
		-
		Z^{\prime}_{k + \frac{n}{2}} \ket{\psi^{\prime}}
	}
	& \leq & \epsilon_2
\\
\label{eq:XZanticommute2old}
	\norm{
		Z^{\prime}_k X^{\prime}_k \ket{\psi^{\prime}}
		+
		X^{\prime}_k Z^{\prime}_k \ket{\psi^{\prime}}
	}
	& \leq & \epsilon_3
\end{eqnarray}
where $k+\frac{n}{2}$ is taken modulo $n$. Then
\begin{enumerate}
\item
for any $s,t \in {\{0,1\}}^{n}$
\begin{equation}
\label{eq:XZanticommutegeneral}
	\norm{
		Z^{\prime t} X^{\prime s} \ket{\psi^{\prime}}
		-
		{(-1)}^{s \cdot t}
		X^{\prime s} Z^{\prime t} \ket{\psi^{\prime}}
	}
	\leq
   O(n^2 \epsilon)
\end{equation}
\item
for any $s \in {\{0,1\}}^{n}$
\begin{equation}
	\norm{
        Z^{\prime s_b s_a} \ket{\psi^{\prime}}
    -
        {(-1)}^{s_a \cdot s_b}
        X^{\prime s} \ket{\psi^{\prime}}
	}
	\leq O(n^2\epsilon)
\end{equation}
\end{enumerate}
where $\epsilon = \max \{ \epsilon_1, \epsilon_2, \epsilon_3\}$.
\end{lemma}

The main difference between Lemma 5 in~\cite{McKague:2015:Selftestingin} and our version here is that we have dropped several requirements on the operators.  Explicit error bounds are the same as in~\cite{McKague:2015:Selftestingin}.

\begin{proof}
First we address conclusion 1.  We will not give the whole proof, but instead refer the reader to the proof of Lemma 5 in~\cite{McKague:2015:Selftestingin} and note the modifications we need to make.  We first note that the proof never uses the fact that any operators are Hermitian or unitary other than to bound the operator.  Hence we can drop these requirements in favour of the less restrictive requirements $\norm{X^{\prime}_k}_\infty \leq 1$ and $\norm{Z^{\prime}_k}_\infty \leq 1$.  The proof does implicitly use the fact that many operators commute, but it is possible to remove these commutations through a careful choice of ordering of the steps.  In particular, for equations (50) to (53) in~\cite{McKague:2015:Selftestingin} the proof starts from small indices and works to larger indices.  Using this ordering requires that the leftmost operators commute to the right, but by starting at highest indices and working to lower indices the operators are already in the correct order and no commutation is required.

Here we also need to point out that if we had chosen a different initial ordering for the operators, for example decreasing indices from left to right, then the proof would still apply, provided that the same ordering is used throughout.

Now we turn our attention to conclusion 2.  Again we will follow the same proof as in~\cite{McKague:2015:Selftestingin}.  We can remove implicit commutations by choosing a different ordering.  However, this would result in the $Z^{\prime}$ operators coming out in the reverse order in the final bound.  To fix this we need to use the fact that all $X^{\prime}$ operators on $A$'s side commute, and also all $Z^{\prime}$ operators on $A$'s side commute.  First, we reverse the ordering of the $X^{\prime}$s on $A$'s side.  Then when these are converted to $Z^{\prime}$'s on $B$'s side, they are reversed again and come out in the correct order.  Meanwhile the $X^{\prime}$'s on $B$'s side are left alone.  They come out in reverse order on $A$'s side, where we commute them around to the correct ordering.

\end{proof}

The conclusions of Lemma~\ref{lemma:acandxz} are nearly the conditions of Lemma~\ref{lemma:graphstateselftestconditions}, so it is straightforward to join the two together.  This will give us our sufficient conditions for self-testing many maximally entangled qubits.

\begin{corollary}
\label{cor:2partysufficientconditions}
Under the conditions of Lemma~\ref{lemma:acandxz}, supposing that ${\{X_k, Z_k\}}_{k=1}^{n}$ are also unitary and Hermitian, there exists a local isometry $\Phi = \Phi_A \otimes \Phi_B$ and a state $\ket{junk}$ such that for every $p,q \in {\{0,1\}}^n$

\begin{equation}
    \norm{
        \Phi(X^{\prime q} Z^{\prime p} \ket{\psi^{\prime}})
        -
        \ket{junk}
        X^q Z^p
        \ket{\psi}
    }
    \leq
    O(n\sqrt{\epsilon})
\end{equation}
where $\ket{\psi}$ is as in Lemma~\ref{lemma:graphstateselftestconditions}.

\end{corollary}

The proof is basically to patch together Lemma~\ref{lemma:graphstateselftestconditions} and~\ref{lemma:acandxz} and work out the bounds.  One aspect which requires a bit more attention is the fact that $\Phi$ is local.  This is a consequence of the construction of $\Phi$ in the proof of Lemma~\ref{lemma:graphstateselftestconditions} and is discussed in~\cite{McKague:2015:Selftestingin}.

%%%%%%%%%%%%%%%%%%%%%%%%%%%%%%%%%%%%%%%%%%%%%%%%%%%%%%%%%%%%%%%%%%%%%%%%%%%%%%%%%%%%%%%%%%%%%%%%%
\section{Parallel CHSH is a self-test}
Now that we have established a framework for self-testing many pairs of maximally entangled qubits, we are ready to develop our self-test based on parallel repetition of CHSH\@.

\subsection{Definition of the non-local game}
The non-local game that we will be using for self-testing is a straightforward parallel repetition of the CHSH game:
\begin{enumerate}
    \item The referee selects question $q \in {\{0,1\}}^{n}$ uniformly at random and sends $q_a$ to Alice and $q_b$ to Bob
    \item Alice and Bob respond with bit strings $x,y \in {\{0,1\}}^{\frac{n}{2}}$.
    \item The referee chooses $k \in \{1 \dots \frac{n}{2} \}$ uniformly at random.
    \item Alice and Bob win if $q_k q_{k + \frac{n}{2}} = x_k \oplus y_k$
\end{enumerate}
There are many possible rules for determining when Alice and Bob win.  The particular rule we use is quite weak in a sense.  We throw away almost all of the information that Alice and Bob provide, and use only one bit from each of them.   However, as we shall see, it is still powerful enough for self-testing the entire state.  The most important feature is that the expected value of this game measures the average CHSH value over all subtests, which could be measured in other ways.

Customarily the value of a CHSH strategy is calculated by assigning $1$ for a win and $-1$ for a loss and summing over the expectations of the four possible questions that can be asked. The value falls between $-4$ and $4$, with quantum strategies between $- 2\sqrt{2}$ and $2 \sqrt{2}$ by the Cirel'son inequality~\cite{Cirelson:1980:Quantum-general}.  We will take the convention of assigning $4$ to a win and $-4$ to a loss, and then taking the expected value over all possible questions and subtests (i.e.\ for various $k = 1 \dots \frac{n}{2}$).  This corresponds to averaging the customary CHSH value over all subtests, and is bounded by $2 \sqrt{2}$ for quantum strategies.

%%%%%%%%%%%%%%%%%%%%%%%%%%%%%%%%%%%%%
\subsection{Modelling the players' behaviour}

In order to use Lemma~\ref{cor:2partysufficientconditions} we will need to define some operators $X^{\prime}_j$ and $Z^{\prime}_j$.  We will do so by way of Alice and Bob's measurement strategies.  But first we need a model.  Alice's behaviour can be modelled as a collection of projective measurements:\footnote{More generally, we could consider POVMs.  However, since the dimension of Alice and Bob's systems are unconstrained, we can instead represent any POVM as projective measurements over a larger Hilbert space.}
\begin{equation}
	\mathcal{M}_{q_a}= {\{ \Pi^{q_a}_a\}}_{a}
\end{equation}
where $q$ is the question, $a \in {\{0,1\}}^{\frac{n}{2}}$ is a string of answers and $\Pi^{q_a}_a\Pi^{q_a}_{a^{\prime}} = \delta_{a,a^{\prime}} \Pi^{q_a}_a$.  We can define projectors for individual symbols in the answer by
\begin{equation}
	\Gamma^{q_a}_{k,x} = \sum_{a: a_k = x} \Pi^{q_a}_a
\end{equation}
where $k \in \{1 \dots \frac{n}{2}\}$, $a_k$ is the $k$-th symbol of $a$, and $x \in \{0,1\}$.  Note that for all $j,k, x, x^{\prime}$ the operators $\Gamma^{q_a}_{j,x}$ and $\Gamma^{q_a}_{k,x^{\prime}}$ commute since $\Pi^q_a$ and $\Pi^q_{a^{\prime}}$ commute for all $a,a^{\prime}$.  We can next define observables for each answer symbol by
\begin{equation}
M^{\prime q_a}_k = \Gamma^{q_a}_{k,0} - \Gamma^{q_a}_{k, 1}
.
\end{equation}
Whenever $q_a = r_a$, $M^{\prime q_a}_k$ and $M^{\prime r_a}_\ell$ will commute by construction.   Now measuring $\mathcal{M}_{q}$ is equivalent to measuring $M^{\prime q}_k$ for each $k$ and returning the resulting eigenvalues as a string (translating 1 to 0 and -1 to 1).  Each $M^{\prime q}_k$ is Hermitian and unitary.

We can model Bob's behaviour in a similar way, defining $N^{\prime q_b}_k$ for $k \in \{\frac{n}{2} + 1 \dots n\}$ analogously to $M^{\prime q_a}_k$.  Then every $N^{\prime q_b}_k$ will commute with every $M^{\prime q_a}_\ell$ since these operators are on different systems.

Let us set
\begin{widetext}
\begin{multline}
   f(q_a, q_b, k) :=
   \bra{\psi^{\prime}}
      \left(
      {(-1)}^{{(q_a)}_k {(q_b)}_k}
      M^{\prime q_a}_k
      N^{\prime q_b}_{k + \frac{n}{2}}
      +
      {(-1)}^{{(\overline{q}_a)}_k {(q_b)}_k}
      M^{\prime \overline{q}_a}_k
      N^{\prime q_b}_{k + \frac{n}{2}}
      \right)
   \ket{\psi^{\prime}}
   +
\\
   \bra{\psi^{\prime}}
      \left(
      {(-1)}^{{(q_a)}_k {(\overline{q}_b)}_k}
      M^{\prime q_a}_k
      N^{\prime \overline{q}_b}_{k + \frac{n}{2}}
      +
      {(-1)}^{{(\overline{q}_a)}_k {(\overline{q}_b)}_k}
      M^{\prime \overline{q}_a}_k
      N^{\prime \overline{q}_b}_{k + \frac{n}{2}}
      \right)
   \ket{\psi^{\prime}}
\end{multline}
\end{widetext}
where $\overline{q}_a$ is $q_a$ with 0 and 1 interchanged, and analogously for $q_b$.

We can interpret $f$ as the CHSH value that Alice and Bob would achieve for subtest $k$ if they use the operators corresponding to questions $q_a$, $q_b$, $\overline{q}_a$ and $\overline{q}_b$.  The value of the non-local game can be expressed as
\begin{equation}
	\label{eq:chshvalue}
   \frac{1}{n 2^{n - 1}}
   \sum_{q_b}
   \sum_{q_a}
   \sum_k
   f(q_a, q_b, k)
   .
\end{equation}
This is the value of Alice and Bob's strategy and is the average CHSH value over all subtests $1 \dots \frac{n}{2}$.  It falls between $- 2 \sqrt{2}$ and $2 \sqrt{2}$ for quantum strategies by the Cirel'son bound~\cite{Cirelson:1980:Quantum-general}.

%%%%%%%%%%%%%%%%%%%%%%%%%%%%%%%%%%%%%%%%%
\subsection{Self-testing from a few questions}

In this section we develop a proof of self-testing in a special case.  Our non-local game has a very large number of questions and it is difficult to make use of the information contained in the corresponding answers.  In particular, each question is asked with an exponentially small probability, so we will know almost nothing about what happens when any particular question is asked.  For now we will ignore this problem and assume that for some particular questions we know what the distribution of answers is.  This will allow us to draw conclusions about the players' behaviour.  In the next section we will return to the full non-local game.

In the previous section we modelled Alice and Bob's measurement strategies.  We next need to define $X^{\prime}_j$ and $Z^{\prime}_j$ so that we can use Corollary~\ref{cor:2partysufficientconditions}.  To do this we will make use of a result from~\cite{McKague:2012:Robust-self-tes}:

\begin{lemma}[McKague et al.~\cite{McKague:2012:Robust-self-tes}, Theorem 2]
\label{lemma:buildxz}
Let $A_0, A_1$ be Hermitian and unitary operators on $\mathcal{H}_A$, $B_0, B_1$ be Hermitian and unitary operators on $\mathcal{H}_B$, and $\ket{\psi} \in \mathcal{H}_A \otimes \mathcal{H}_B$ be a state such that
\begin{equation}
    \bra{\psi}
        \left(
            A_0 B_0
            +
            A_1 B_0
            +
            A_0 B_1
            -
            A_1 A_1
        \right)
    \ket{\psi}
    \geq 2 \sqrt{2} - \delta
\end{equation}
for $0 \leq \delta \leq 1$.  Then setting
\begin{eqnarray}
    X^{\prime}_A & = & A_0 \\
    Z^{\prime}_A & = & A_1 \\
    X^{\prime}_B & = & \frac{B_0 + B_1}{| B_0 + B_1|} \\
    Z^{\prime}_B & = & \frac{B_0 - B_1}{| B_0 - B_1|}
\end{eqnarray}
we have
\begin{eqnarray}
    \norm{
        X^{\prime}_A
        Z^{\prime}_A
        \ket{\psi}
        +
        Z^{\prime}_A
        X^{\prime}_A
        \ket{\psi}
    } & \leq & 4{(\delta \sqrt{2})}^\frac{1}{2} \\
    \norm{
        X^{\prime}_B
        Z^{\prime}_B
        \ket{\psi}
        +
        Z^{\prime}_B
        X^{\prime}_B
        \ket{\psi}
    } & \leq & 4{(\delta \sqrt{2})}^\frac{1}{2} \\
    \norm{
        X^{\prime}_A
        \ket{\psi}
        -
        Z^{\prime}_B
        \ket{\psi}
    } & \leq & 4 {(\delta \sqrt{2})}^{\frac{1}{4}}\\
    \norm{
        X^{\prime}_B
        \ket{\psi}
        -
        Z^{\prime}_A
        \ket{\psi}
    } & \leq & 4 {(\delta \sqrt{2})}^{\frac{1}{4}}
\end{eqnarray}
where
\begin{equation}
    |M| := \sqrt{M^2}
\end{equation}
for some operator $M$.  Further, $X^{\prime}_A$, $X^{\prime}_B$, $Z^{\prime}_A$ and $Z^{\prime}_B$ are all Hermitian and unitary.
\end{lemma}

Note that it is possible that zero is an eigenvalue of $|M|$.  For theses cases we must change $M$ on the associated subspace to remove zero eigenvalues and avoid divide-by-zero problems.  We can make the new eigenvalue as close to zero as we like, which corresponds to a very small change in Bob's measurements.  By making this change small enough, it will be negligible compared to the other errors in the system.

In the next lemma we show how to link the above result with Corollary~\ref{cor:2partysufficientconditions} to obtain a self-testing result.  Here we assume that we have very good information about a small number of possible questions.  This could be accomplished, for example, by having the referee ask only a subset of possible questions.

\begin{lemma}
\label{lemma:questionstoselftesting}
Suppose that:

\begin{enumerate}
   \item For each $k \in \{1 \dots \frac{n}{2}\}$
      \begin{equation}
         f(0 \dots 0, 0 \dots 0, k) \geq 2 \sqrt{2} - \delta
      \end{equation}

   \item
      For every $u,v \in \{1 \dots \frac{n}{2}\}$ with $u \neq v$ there exists a question $q_a$ with ${(q_a)}_u = 0$ and ${(q_a)}_v = 1$ such that
      \begin{equation}
         f(q_a, 0 \dots 0, j) \geq 2 \sqrt{2} - \delta
      \end{equation}
      for $j = u,v$.
\end{enumerate}
Then there exist operators $X^{\prime}_k$ and $Z^{\prime}_k$ for each $k \in \{1 \dots \frac{n}{2}\}$ such that the conditions of Corollary~\ref{cor:2partysufficientconditions} are satisfied with
$\epsilon_1 = 32{(\delta\sqrt{2})}^{\frac{1}{4}}$,
$\epsilon_2= 4{(\delta\sqrt{2})}^{\frac{1}{4}}$
and
$\epsilon_3 = 4{(\delta\sqrt{2})}^{\frac{1}{2}}$.
\end{lemma}

\begin{proof}
For each $k$, the conditions of Lemma~\ref{lemma:buildxz} are satisfied by setting
\begin{eqnarray}
	A_0 & = & M^{\prime 0\dots 0}_k \\
	A_1 & = & M^{\prime 1\dots 1}_k \\
	B_0 & = & N^{\prime 0\dots 0}_{k + n/2} \\
	B_1 & = & N^{\prime 1\dots 1}_{k + n/2}
\end{eqnarray}
Hence on Alice's side ($k = 1 \dots \frac{n}{2}$) we set
\begin{eqnarray}
    X^{\prime}_k & := & M^{\prime 0\dots 0}_k
\\
    Z^{\prime}_k & := & M^{\prime 1\dots 1}_k
\end{eqnarray}
for $k = 1 \dots \frac{n}{2}$.  Note that by construction all $Z^{\prime}_k$ pairwise commute on Alice's side, as do all $X^{\prime }_k$.  For Bob's side ($k= \frac{n}{2} + 1 \dots n$)  we define
\begin{eqnarray}
   X^{\prime}_k & := &
   \frac{
      N^{\prime 0 \dots 0}_k + N^{\prime 1 \dots 1}_k
   }{
      \left|
         N^{\prime 0 \dots 0}_k + N^{\prime 1 \dots 1}_k
      \right|
   }
\\
   Z^{\prime}_k & := &
   \frac{
      N^{\prime 0 \dots 0}_k - N^{\prime 1 \dots 1}_k
   }{
      \left|
         N^{\prime 0 \dots 0}_k - N^{\prime 1 \dots 1}_k
      \right|
   }.
\end{eqnarray}
Here it is possible that two $X^{\prime}_j$'s on Bob's side do not commute.  With these definitions we can set $\epsilon_2 =4 {(\delta \sqrt{2})}^{\frac{1}{4}}$ and $\epsilon_3 = 4{(\delta\sqrt{2})}^{\frac{1}{2}}$ for Corollary~\ref{cor:2partysufficientconditions}.

Now fix $k, \ell$ with $k \neq \ell$.  We want to know that $Z^{\prime}_\ell$ approximately commutes with $X^{\prime}_k$ to fix $\epsilon_1$.  If $k$ is on Alice's side and $\ell$ is on Bob's side (or vice versa) then the operators commute exactly.  Hence the two cases to check are where $k$ and $\ell$ are both at most $\frac{n}{2}$ or when they are both larger than $\frac{n}{2}$.

Supposing that $k$ and $\ell$ are both larger than $\frac{n}{2}$, we use condition 2 with $u = \ell - \frac{n}{2}$ and $v = k - \frac{n}{2}$  to find $q_a$ and define
\begin{eqnarray}
   X^{\prime \prime}_{\ell - \frac{n}{2}} & := & M^{\prime q_a}_{\ell - \frac{n}{2}}
\\
   Z^{\prime \prime}_{k - \frac{n}{2}} & := & M^{\prime q_a}_ {k - \frac{n}{2}}
   .
\end{eqnarray}
Applying Lemma~\ref{lemma:buildxz} twice gives
\begin{eqnarray}
   \norm{
      X^{\prime \prime}_{\ell - \frac{n}{2}}
      \ket{\psi}
      -
      Z^{\prime}_{\ell}
      \ket{\psi}
   }
   & \leq &
   4 {(\delta \sqrt{2})}^{\frac{1}{4}}
\\
   \norm{
      Z^{\prime \prime}_{k - \frac{n}{2}}
      \ket{\psi}
      -
      X^{\prime}_{k}
      \ket{\psi}
   }
   & \leq &
   4 {(\delta \sqrt{2})}^{\frac{1}{4}}
\end{eqnarray}
Here Bob's operators have not changed from their previous definition.  Moreover, since $X^{\prime \prime}_{\ell - \frac{n}{2}}$ and $Z^{\prime \prime}_{k - \frac{n}{2}}$ are defined from the same question, they commute with each other.  Now it is a straightforward exercise to combine these two estimates and show that for this case we can use $\epsilon_1 = 16 {(\delta \sqrt{2})}^{\frac{1}{4}}$.
Roughly, we use the estimates to move $X^{\prime}_k$ and $Z^{\prime}_\ell$ over to Alice's side as
$X^{\prime \prime}_{\ell - \frac{n}{2}}$ and $Z^{\prime \prime}_{k - \frac{n}{2}}$,
where we can commute them, and then move them back to Bob's side.

For the other case, where $k$ and $\ell$ are less than $\frac{n}{2}$ so that $X^{\prime}_k$ and $Z^{\prime}_\ell$ are on Alice's side, we first move $X^{\prime}_k$ and $Z^{\prime}_\ell$ over to Bob's side, and then back to Alice side, but as $X^{\prime \prime}_k$ and $Z^{\prime \prime}_\ell$, which commute.  Then they go back to
$X^{\prime}_k$ and $Z^{\prime}_\ell$
via Bob's side.  Here we end up with a worse estimate, and we finally set
\begin{equation}
   \epsilon_1 = 32 {(\delta \sqrt{2})}^{\frac{1}{4}}.
\end{equation}

\end{proof}

Although this is not our main result, we now have a type of parallel CHSH self-test.  We simply need to estimate the values (for all $k$) of $f(0\dots 0, 0 \dots 0, k)$ and $f(q_a, 0 \dots 0, k)$ for a suitable set of $q_a$.  One minimal set of such $q_a$ would be the strings indexed by $j$ where the $k$th bit is 1 exactly when the $j$th bit of the binary representation of $k$ is 1.  The total number of questions is then $O(\log n)$.  Whenever $k \neq \ell$ their binary representations will differ in some position $j$.  Then the question with index $j$ will have bits $k$ and $\ell$ not equal.  This is analogous to the parallel Mayers-Yao test developed in~\cite{McKague:2015:Selftestingin}.

The test is in a strict sense not parallel since the question for one subtest is not independent of the questions for other subtests.  Nevertheless, it does provide a self-test and the small number of questions could be useful for some applications.

%%%%%%%%%%%%%%%%%%%%%%%%%%%%%%%%%%%%%%%%%%
\subsection{Self-testing from parallel CHSH}

In the previous section we showed self-testing, but only looked at a very small number of possible questions.  In the full non-local game we would not be able to draw any strong conclusions about specific questions since they would be asked with very small probability.  In this section we take a different approach.  If the players' value of the non-local game is very high then there must be at least \emph{some} questions which are driving the value of the non-local game up.  These might not be the questions used in Lemma~\ref{lemma:questionstoselftesting}, but through various symmetries of the non-local game we can expand the list of questions which work for Lemma~\ref{lemma:questionstoselftesting}.

First let us discuss symmetries.  Lemma~\ref{lemma:questionstoselftesting} requires that questions $0 \dots 0$ and $1 \dots 1$ for Alice and Bob can be used to achieve a high CHSH value.  However, there is obvious symmetry in the non-local game and it is straightforward to relabel the questions and answers so that many different sets of questions can be used.  In particular, the mapping
\begin{eqnarray}
   q_a & \mapsto & q_a \oplus 1_k  \\
   M^{\prime q_a}_\ell & \mapsto & M^{\prime q_a \oplus 1_k}_\ell, \, \ell = 1 \dots \frac{n}{2} \\
   N^{\prime q_b}_{k + n/2} & \mapsto & {(-1)}^{{(q_b)}_{k}} N^{\prime q_b}_{k + n/2}, \, k = 1 \dots \frac{n}{2}
\end{eqnarray}
allows us to a flip bit $k$ in all questions on Alice's side.  An analogous map allows us to flip any bit in all questions on Bob's side.  These mappings preserve the value of the non-local game. Applying a suitable set of them allows us to map any $q_a$  and $q_b$ to $0 \dots 0$, necessarily taking $\overline{q_a}$ and $\overline{q_b}$ to $1 \dots 1$.

Now we attempt to find some questions which give $f(q_a, q_b, k)$ a large value.  The basic idea is that it is impossible for the value of every sub-test to be below average.  Suppose that the value of the non-local game is at least $2 \sqrt{2} - \epsilon$.  We can view the value of the non-local game given in~\eqref{eq:chshvalue} as an average over possible $q_b$ of some value depending on $q_b$.  At least one $q_b$ is at least average, so there exists some $q_b$ such that
\begin{equation}
   \frac{1}{n 2^{\frac{n}{2} - 1}}
   \sum_{q_a}
   \sum_k
   f(q_a, q_b, k)
   \geq
   2 \sqrt{2} - \epsilon
   .
\end{equation}
Let us choose such a $q_b$ and then remap the non-local game so that $q_b$ maps to $0 \dots 0$.  We can now do the same trick with $q_a$ to find that there is at least one $q_a$ such that
\begin{equation}
   \frac{2}{n}
   \sum_k
   f(q_a, 0 \dots 0, k)
   \geq
   2 \sqrt{2} - \epsilon
   .
\end{equation}
Again we remap the non-local game so that this above average $q_a$ gets mapped to $0 \dots 0$.

In order to find a bound for each separate $k$, we observe that for each fixed $k$ the summand is bounded above by $2 \sqrt{2}$ using the Cirel'son inequality.  Supposing that all the error falls on one value of $k$, we find
\begin{equation}
   f(0 \dots 0, 0 \dots 0, k) \geq 2 \sqrt{2} - \frac{n}{2} \epsilon
\end{equation}
for all $k$.  Hence we have satisfied the first condition of Lemma~\ref{lemma:questionstoselftesting}.

Turning our attention to the second condition of Lemma~\ref{lemma:questionstoselftesting}, we first use a similar reasoning to above to find that for each $j$
\begin{equation}
   \frac{1}{2^{\frac{n}{2}}}
   \sum_{q_a}
   f(q_a, 0\dots 0, j)
   \geq
   2 \sqrt{2} - \frac{n}{2} \epsilon
   .
\end{equation}
Now let us fix $k, \ell$ and partition the possible $q_a$ into two sets.  The first set $S$ has either $({(q_a)}_k, {(q_a)}_\ell) = (0,0)$ or $({(q_a)}_k, {(q_a)}_\ell) = (1,1)$.  That is, questions in $S$ have the $k$th and $\ell$th bits equal.  The second set $T$ is the remaining $q_a$.  Note that sets are of equal size, and $q_a$ and $\overline{q}_a$ are in the same set.  Supposing that all the error falls on $T$ we find
\begin{equation}
   \frac{1}{2^{\frac{n}{2}}}
   \sum_{q_a \in T}
   f(q_a, 0\dots 0, j)
   \geq
   2 \sqrt{2} - n \epsilon
   .
\end{equation}
for $j = k, \ell$.  Now at least one $q_a$ is at least average, so there is some $q_a$ which differs in bits $k$ and $\ell$ which satisfies
\begin{equation}
   f(q_a, 0\dots 0, j)
   \geq
   2 \sqrt{2} - n \epsilon
\end{equation}
for $j = k, \ell$.  From the definition of $f$, we see that the above must also hold for $\overline{q}_a$.  At least one of $q_a$ and $\overline{q}_a$ has the $k$th bit equal to 0 and the $\ell$th bit equal to 1.  So, we have satisfied the second condition of Lemma~\ref{lemma:questionstoselftesting}.

We have shown that if the value of the non-local game is high enough then the conditions of Lemma~\ref{lemma:questionstoselftesting} are satisfied with
$
   \delta = O(n \epsilon)
$,
proving our main result.

\begin{theorem}
If the value of non-local game for $\frac{n}{2}$ copies of CHSH achieved by a state
$\ket{\psi^{\prime}} \in \mathcal{H}_A \otimes \mathcal{H}_B$
is at least $2 \sqrt{2} - \epsilon$ then there exist
\begin{itemize}
   \item Hermitian and unitary $X^{\prime}_k$ and $Z^{\prime}_k$, defined on $\mathcal{H}_A$ for $k \leq \frac{n}{2}$ and $\mathcal{H}_B$ for $k > \frac{n}{2}$
   \item a local isometry
      $\Phi :
         \mathcal{H}_A \otimes \mathcal{H}_B
         \rightarrow
         \mathcal{H}_A \otimes \mathcal{H}_B
         \otimes
         \mathcal{H}_2 \otimes \mathcal{H}_2
      $
   \item a state $\ket{junk} \in \mathcal{H}_A \otimes \mathcal{H}_B$
\end{itemize}
such that for any $p,q \in {\{0,1\}}^n$
\begin{equation}
   \norm{
      \Phi\left(
         X^{\prime q}
         Z^{\prime p}
         \ket{\psi^{\prime}}
      \right)
      -
      \ket{junk}
      X^{q}
      Z^{p}
      \ket{\psi}
   }
   \leq O(n^{\frac{9}{8}}\epsilon^{\frac{1}{8}})
	.
\end{equation}
where $\ket{\psi}$ is $\frac{n}{2}$ pairs of maximally entangled qubits.
\end{theorem}

%%%%%%%%%%%%%%%%%%%%%%%%%%%%%%%%%%%%%%%%%%%%%%%%%%%%%%%%%%%%%%%%%%%%%%%%%%%%%%%%%%%%%%%%%%%%%%%%%
\section{Discussion}

We have shown that the CHSH game can be used for testing many pairs of maximally entangled qubits, and that the error bound scales polynomially.  This opens the possibility of using CHSH in more applications.  The two constructions allow for either a strictly parallel CHSH test, where all subtests are independent, or a construction which requires uses a logarithmic number of questions.

Coladangelo~\cite{Coladangelo:2016:Parallelselftesting} achieved a similar result to ours, allowing parallel self-testing using CHSH.\@  Their result is more general in that they considers the tilted CHSH inequality as well, which allows for testing of non-maximally entangled pairs of qubits.  Robustness is $O(n^{3/2} \epsilon^{1/2})$ as compared to our $O(n^{9/8} \epsilon^{1/8})$.  While their scaling in $\epsilon$ is better, ours is better in the scaling in $n$.

The parallel CHSH self-test could be used in device independent quantum key distribution to allow for parallel protocols.  Device independent quantum key distribution has been achieved for serial repetition~\cite{Vazirani:2014:Fullydeviceindependent} but not for parallel repetition.  The parallel self-test variant with a logarithmic number of questions might be useful in device independent quantum random number generators, where the smaller number of questions could reduce the about of seed randomness required.  See~\cite{Miller:2014:Robustprotocolssecurely} for an example of a quantum random number generator which utilises self-testing.  Finally, parallel self-tests could simplify protocols for interactive proofs, allowing for a smaller number of rounds of communication as compared to serial repetition for constructions as in~\cite{Reichardt2013:Classicalcommandofquantumsystems}, or a reduced number of provers for constructions as in~\cite{McKague:2013:Interactive-proofs-for-BQP-via-self-tested-graph-states} or~\cite{Hajdusek::Device-IndependentVerifiableBlindQuantumComputation}.

\begin{acknowledgements}
	  This work is partially funded by the Dodd-Walls Centre for Photonic and Quantum Technologies.
\end{acknowledgements}

\let\OLDthebibliography\thebibliography
\renewcommand\thebibliography[1]{
  \OLDthebibliography{#1}
  \setlength{\parskip}{0pt}
  \setlength{\itemsep}{0pt plus 1.6ex}
}

\bibliographystyle{halphamm}
\renewcommand{\doi}[1]{\href{https://doi.org/\detokenize{#1}}{DOI: \detokenize{#1}}}

\bibliography{Global_Bibliography}

\end{document}